\documentclass[a4paper,UKenglish,cleveref, autoref]{lipics-v2019}

\title{Reasoning about Quality and Fuzziness of Strategic Behaviours}

\author{Patricia~Bouyer}{LSV, CNRS \& ENS Paris-Saclay, Univ. Paris-Saclay, France}{}{}{}
\author{Orna~Kupferman}{Hebrew University, Israel}{}{}{}
\author{Nicolas~Markey}{Irisa, CNRS \& Inria \& Univ. Rennes, France}{}{}{}
\author{Bastien~Maubert}{Universit\'a degli Studi di Napoli ``Federico~II'', Italy}{}{}{}
\author{Aniello~Murano}{Universit\'a degli Studi di Napoli ``Federico~II'', Italy}{}{}{}
\author{Giuseppe~Perelli}{University of Leicester, UK}{}{}{}

\authorrunning{Anonymous}

\Copyright{}


\category{}

\relatedversion{}

\supplement{}

\acknowledgements{Perelli thanks the support of the project ``dSynMA'',
  funded by the ERC under the European Union's Horizon 2020 research
  and innovation programme (grant agreement No 772459).  Bouyer and
  Markey thank the support of the ERC Project EQualIS (308087).}

\hideLIPIcs  

\EventEditors{}
\EventNoEds{2}
\EventLongTitle{}
\EventShortTitle{}
\EventAcronym{CVIT}
\EventYear{2016}
\EventDate{December 24--27, 2016}
\EventLocation{Little Whinging, United Kingdom}
\EventLogo{}
\SeriesVolume{42}
\ArticleNo{23}

\usepackage{amsmath, amssymb, amsfonts, amsthm, stmaryrd}
\usepackage{pgf}
\usepackage{tikz}
\usetikzlibrary{arrows,automata, positioning}
\usepackage{bwc}
\usepackage{xspace}
\usepackage{todonotes}
\usepackage{colonequals}

\nolinenumbers

\usepackage{multidef}

\usepackage{mathstyle}

\usepackage{colored-comments,macros,macros-giuseppe}

\begin{document}

\maketitle

\begin{abstract}
\newcommand\change[1]{{\color{magenta} #1}}

\emph{Temporal logics} are extensively used for the specification of
on-going behaviours of reactive systems. Two significant developments
in this area are the extension of traditional temporal logics with
modalities that enable the specification of on-going \emph{strategic}
behaviours in multi-agent systems, and the transition of temporal
logics to a \emph{quantitative} setting, where different satisfaction
values enable the specifier to formalise concepts such as certainty or
quality. We introduce and study \emph{\FSL}---a~quantitative extension
of~\SL (\emph{Strategy Logic}), one of the most natural and expressive
logics describing strategic behaviours. The~satisfaction value of an
$\FSL$ formula is a real value in~$[0,1]$, reflecting ``how much'' or
``how well''
the strategic on-going objectives of
the underlying agents are satisfied.  
We demonstrate the applications of \FSL in quantitative reasoning
about multi-agent systems, by showing how it can express concepts of
stability in multi-agent systems, and how it generalises some fuzzy
temporal logics. We also provide a model-checking algorithm for our
logic, based on a quantitative extension of Quantified \CTLs.



\end{abstract}

\section{Introduction}
\label{sec-intro}
One of the significant developments in formal reasoning has been the
use of \emph{temporal logics} for the specification of on-going
behaviours of reactive systems~\cite{Pnu81,CE81,EH86}. Traditional
temporal logics are interpreted over Kripke structures, modelling
closed systems, and can quantify the computations of the systems in a
universal and existential manner. The~need to reason about
multi-agents systems has led to the development of specification
formalisms that enable the specification of on-going strategic
behaviours in multi-agent
systems~\cite{AHK02,CHP10,MMPV14}. Essentially,
these formalisms, most notably \ATL, \ATLs, and Strategy Logic~(\SL),
include quantification of strategies of the different agents and of
the computations they may force the system into, making it possible to
specify concepts that have been traditionally studied in game theory.

The duration of games in game theory is typically finite and the outcome of the game depends on its final position \cite{NR99,NRTV07}.
In contrast, agents in multi-agent systems maintain an \emph{on-going interaction} with 
each other~\cite{HP85}, and reasoning about their behaviour refers
not to their final state (in~fact, we consider non-terminating
systems, with no final state) but rather to the \emph{language} of
computations that they generate.  While \SL, which subsumes \ATLs enables the specification
of rich on-going strategic behaviours, its semantics is Boolean:
a~system may satisfy a specification or it may~not.  The~Boolean
nature of traditional temporal logic is a real obstacle in the context
of strategic reasoning. Indeed, while many strategies may attain a
desired objective, they may do so at different levels of quality or
certainty. Consequently, designers would be willing to give up manual
design only after being convinced that the automatic procedure that
replaces it generates systems of comparable quality and certainty. For
this to happen, one should first extend the specification formalism to
one that supports quantitative aspects of the systems and the
strategies.

The logic $\FLTL$ is a multi-valued logic that augments \LTL with
quality operators~\cite{ABK16}. The satisfaction value of an $\FLTL$
formula is a real value in $[0,1]$, where the higher the value, the
higher the quality in which the computation satisfies the
specification.  The~quality operators in~$\Func$ can prioritise different
scenarios or reduce the satisfaction value of computations in which
delays occur.  For example, the~set~$\Func$ may contain the
$\min\{x,y\}$, $\max\{x,y\}$, and $1-x$ functions, which are the
standard quantitative analogues of the $\wedge$, $\vee$, and $\neg$
operators. The~novelty of \FLTL is the ability to manipulate values by
arbitrary functions. For~example, $\Func$~may contain the
weighted-average function $\avg{\lambda}$. The satisfaction value of
the formula $\psi_1\avg{\lambda}\psi_2$ is the weighted (according
to~$\lambda$) average between the satisfaction values of $\psi_1$
and~$\psi_2$. This enables the specification of the quality of the
system to interpolate different aspects of~it.  As~an example,
consider the \FLTL formula $\G({\it req} \rightarrow ({\it
  grant} \avg{\frac{2}{3}} \X {\it grant}))$. The formula
states that we want requests to be granted immediately and the grant
to hold for two transactions. When this always holds, the satisfaction
value is $\frac{2}{3}+\frac{1}{3}=1$. We are quite okay with grants
that are given immediately and last for only one transaction, in which
case the satisfaction value is $\frac{2}{3}$, and less content when
grants arrive with a delay, in which case the satisfaction value is
$\frac{1}{3}$.

We introduce and study \FSL: an analogous multi-valued extension
of~\SL.  In addition to the quantitative semantics that arises from the
functions in $\Func$, another important element of \FSL is that its
semantics is defined with respect to \emph{weighted} multi-agent
systems, namely ones where atomic propositions have truth values in
$[0,1]$, reflecting quality or certainty. Thus, a model-checking procedure for \FSL, which is our main contribution,  enables formal reasoning about both quality and fuzziness of strategic behaviours.

As a motivating example, consider security drones that may patrol
different height levels.  Using \FSL, we~can specify, for example
(see~specific formulas in Section~\ref{ssec-expressiveness}), the
quality of strategies for the drones whose objectives are to fly above
and below all uncontrollable drones and perform certain actions when
uncontrollable drones exhibit some disallowed behaviour. Indeed, the
multi-valued atomic propositions are used to specify the different
heights, temporal operators are used for specifying on-going
behaviours, the~functions in~$\Func$ may be used to refer to these
behaviours in a quantitative manner, for example to compare heights
and to specify the satisfaction level that the designer gives to
different possible scenarios. Note that the \FSL formula does not
specify the ability of the drones to behave in some desired
manner. Rather, it~associates a satisfaction value in $[0,1]$ with each
behaviour. This suggests that \FSL can be used not only for a
quantitative specification of strategic behaviours but also for
quantizing notions from game theory that are traditionally
Boolean. For example (see specific formulas in
Section~\ref{ssec-expressiveness}), beyond specifying that the agents
are in a Nash Equilibrium, we can specify how far they are from an
equilibrium, namely how much an agent may gain by a deviation that
witnesses the instability. \bmchanged{As a result we can express concepts such as
$\epsilon$-Nash Equilibria~\cite{NRTV07}.} 

The logic \FSL enables the quantification of strategies for the
agents. We show that the quantification of strategies can be reduced
to a Boolean quantification of atomic propositions, which enables us
to reduce the model-checking problem of \FSL to that of \BQCTLsf~--~an
extension of \CTLsf~\cite{ABK16} with quantified Boolean atomic
propositions.  A general technique for \CTLs model-checking algorithms
is to repeatedly evaluate the innermost state subformula by viewing it
as an (existentially or universally quantified) \LTL formula, and add
a fresh atomic proposition that replaces this
subformula~\cite{EL87}. This general technique is applied also to
\CTLsf, with the fresh atomic propositions being weighted
\cite{ABK16}. For \BQCTLsf formulas, however, one cannot apply
it. Indeed, the externally quantified atomic propositions may appear
in different subformulas, and we cannot evaluate them one by one
without fixing the same assignment for the quantified atomic
propositions. Instead, we extend the automaton-theoretic approach to
\CTLs model-checking \cite{KVW00} to handle quantified propositions:
given a \BQCTLsf formula $\psi$ and predicate $P \subseteq [0,1]$, we
construct an alternating parity tree automaton that accepts exactly
all the weighted trees $t$ such that the satisfaction value of $\psi$
in $t$ is in $P$. The translation, and hence the complexity of the
model-checking problem, is non-elementary. \bmchanged{More precisely we show that
it is \kEXPTIME[(k+1)]-complete for formulas involving at most $k$
quantifications on atomic propositions, and we show a similar
complexity result for \FSL, in terms of nesting of strategy quantifiers.}

\subparagraph{Related works.}
There have been long lines of works about games with quantitative
objectives (in~a broad sense), \textit{e.g.}~stochastic games~\cite{Sha53, FV97}, timed
games~\cite{AMPS98}, or weighted games with various kinds of objectives
(parity~\cite{EJ91}, mean-payoff~\cite{EM79} or
energy~\cite{CAHS03,BFLMS08}).  This does not limit to zero-sum games,
but also includes the study of various solution concepts (see for instance~\cite{Umm10,BMR14,BBD14,BPRS16,AKP18}).
%
Similarly, extensions of the classical temporal logics \LTL and \CTL
with quantitative semantics have been studied in different contexts,
with discounting~\cite{AFHMS05,ABK14}, averaging~\cite{BDL12,BMM14},
or richer constructs~\cite{BCHK14,ABK16}.
In~contrast, the study of quantitative temporal logics for strategic
reasoning has remained rather limited: works on \FLTL include
algorithms for synthesis and rational
synthesis~\cite{ABK16,AK16,AKP18,AKRV17}, but no logic combines the
quantitative aspect of \FLTL with the strategic reasoning offered
by~\SL. Thus, to the best of our knowledge, our model-checking
algorithm for \FSL is the first decidability result for a quantitative
extension of a strategic specification formalism (without restricting
to bounded-memory strategies).

Baier and others have focused on a variant of \SL in a stochastic
setting~\cite{BBGK12}; model checking was
proven decidable for memoryless strategies, and undecidable in the
general case. A~quantitative version of \SL with Boolean goals over
one-counter games has been considered in~\cite{BGM15}; only
a periodicity property was proven, and no model-checking algorithm is
known in that setting as well. Finally, Graded
\SL~\cite{AMMR18} extends \SL by quantifying on the number of
strategies witnessing a given strategy quantifier, and is decidable.

The other quantitative extensions we know of concern \ATL\slash \ATL*,
and most of the results are actually adaptations of similar
(decidability) results for the corresponding extensions of \CTL
and~\CTL*; this includes probabilistic \ATL~\cite{CL07},
timed \ATL~\cite{HP06,BLMO07}, multi-valued~\ATL~\cite{JKP16}, and
weighted versions of
\ATL~\cite{LMO06,BG13,Ves15}.  Finally, some works have considered
non-quantitative \ATL with quantitative constraints on the set of
allowed strategies~\cite{ALNR10,DM18}, proving decidability of the
model-checking problem.



\section{Quantitative Strategy Logic}
\label{sec-slf}


Let $\Sigma$ be an alphabet. A~\emph{finite} (resp. \emph{infinite})
\emph{word} over $\Sigma$ is an element of $\Sigma^{*}$
(resp.~$\Sigma^{\omega}$).  The \emph{length} of a finite word
$w=w_{0}w_{1}\ldots w_{n}$ is $|w|\egdef n+1$, and $\last(w)\egdef
w_{n}$ is its last letter.  Given a finite (resp. infinite) word $w$
and $0 \leq i < |w|$ (resp. $i\in\SetN$), we let $w_{i}$ be the letter
at position $i$ in $w$, $w_{\leq i}=w_0\ldots w_i$ is the (nonempty) prefix of
$w$ that ends at position $i$ and
$w_{\geq i}=w_iw_{i+1}\ldots$ is the suffix of $w$ that starts at
position~$i$. As~usual, for any partial function~$f$,
we~write~$\dom{f}$ for the domain of~$f$.

Strategy logic with functions, denoted \FSL, generalises both
\SL~\cite{CHP10,MMPV14} and \FLTL~\cite{ABK16} by replacing the
Boolean operators of \SL with arbitrary functions over $[0,1]$.  The
logic is actually a family of logics, each parameterised by a set
$\Func$ of functions.

\subsection{Syntax}
\label{sec-syntax-SLf}

We build on the branching-time variant of
\SL~\cite{DBLP:conf/csl/FijalkowMMR18}, which does not add
expressiveness with respect to the classic semantics~\cite{MMPV14} but
presents several benefits (see~\cite{DBLP:conf/csl/FijalkowMMR18} for
more details), one of them being that it makes the connection with
Quantified \CTL tighter.

\begin{definition}[Syntax]
  Let $\Func\subseteq \{f\colon [0,1]^m\to [0,1] \mid m \in \SetN \}$
  be a set of functions over $[0,1]$ of possibly different arities.
  The syntax of \FSL is defined with respect to a finite set of
  \emph{atomic propositions}~$\AP$, a~finite set of agents $\Agt$ and a set of
  strategy variables~$\Var$.
  The set of \FSL formulas is defined by the following grammar:
\begin{align*}
  	\varphi & \coloncolonequals p \mid \estrat\varphi \mid (\ag, x)  \varphi 
                 \mid \A \psi 
        \mid f(\varphi, \ldots, \varphi)\\
	\psi & \coloncolonequals \varphi \mid \X \psi \mid \psi \U \psi
        \mid f(\psi, \ldots, \psi)
\end{align*}
where $p\in\AP$,  $x\in\Var$, $\ag\in\Agt$, and $f\in \Func$.
\end{definition}
Formulas of type $\varphi$ are called \emph{state formulas}, those of type
$\psi$ are called \emph{path formulas}. 
Formulas $\estrat\varphi$ are called \emph{strategy quantifications}
whereas formulas $(\ag, x)  \varphi$ are called
\emph{bindings}. Modalities $\X$ and $\U$ are temporal modalities,
which take a specific quantitative semantics as we see below.

We may use $\top$, $\vee$, and $\neg$ to denote functions~$1$, $\max$ and
$1 - x$, respectively.
We can then define the following 
classic abbreviations: ${\perp\colonequals\neg\top}$, ${\varphi\wedge\varphi'
\colonequals \neg (\neg \varphi \vee \neg \varphi')}$, ${\varphi\rightarrow \varphi' \colonequals \neg
\varphi \vee \varphi'}$, ${\F\psi \colonequals \top \U \psi}$,
${\G\psi \colonequals \neg \F \neg \psi}$ and  ${\astrat\varphi\colonequals \neg\estrat\neg\varphi}$.

  Intuitively, the value of formula $\varphi
\vee \varphi'$ is the maximal value of the two formulas $\varphi$ and
$\varphi'$,  $\varphi \wedge \varphi'$ takes the minimal value
of the two formulas, and the value of $\neg \varphi$ is one minus that
of $\varphi$. The implication $\varphi \rightarrow \varphi'$
thus takes the maximal value between that of $\varphi'$ and one minus that
 of $\varphi$. 

In a Boolean setting, we assume that the values of the atomic
propositions are in $\{0,1\}$ ($0$ represents false whereas $1$
represents true), and so are the output values of functions in
$\Func$.  One can then check that $\varphi \vee \varphi'$, $\varphi
\wedge \varphi'$, $\neg\varphi$ and $\varphi \rightarrow \varphi'$ take their usual
Boolean meaning.

We will come back later to temporal modalities,  strategy
quantifications and bindings. 

\subsection{Semantics}

While \SL is evaluated on
classic concurrent game structures with Boolean valuations for atomic propositions,
  \FSL formulas are interpreted 
on weighted concurrent game structures, in which atomic propositions
have values in $[0,1]$, and that we now present.

\begin{definition}
  \label{def-wcgs}
  A \emph{weighted concurrent game structure}
  (WCGS) is a tuple $\calG =
  \tuple{\AP,\Agt,\Act,\Pos,\pos_\init,\Delta,\labels}$ where
 $\AP$ is a finite set of
  atomic propositions,    $\Agt$ is a finite set of agents,
  $\Act$~is a
  finite set of actions, $\Pos$~is a finite set of states,
  $\pos_\init \in \Pos$ is an initial state, $\Delta\colon \Pos
  \times\Act^\Agt\to \Pos$ is the transition function, and
  $\labels\colon \Pos\to[0,1]^{\AP}$ is a weight
  function. 
\end{definition}

An~element of $\Act^\Agt$ is a \emph{joint action}.  For ${\pos \in \Pos}$, we let $\succs(\pos)$ be the set
$\{\pos'\in \Pos\mid \exists \jact\in\Act^\Agt.\ \pos'=\Delta(\pos,\jact)\}$. For~the sake
of simplicity, we~assume in the sequel that
$\succs(\pos)\not=\varnothing$ for all~$\pos\in \Pos$.

A~\emph{play} in~$\calG$ 
is an
infinite sequence $\pi=(\pos_i)_{i \in \bbN}$ of states in~$\Pos$ such that
$\pos_0=\pos_\init$ and $\pos_i\in\succs(\pos_{i-1})$ for all~$i>0$.  We~write
$\Play_{\calG}$ for the set of plays in~$\calG$, and
$\Play_{\calG}(v)$ for the set of plays in~$\calG$ starting from~$v$.
In this
and all similar notations, we might omit to mention~$\calG$ when it is
clear from the context. 
A~(strict) prefix of a play~$\pi$ is a finite sequence
$\rho=(\pos_i)_{0\leq i\leq L}$, for some~$L\in\bbN$, which we denote
$\pi_{\le L}$. We~write
$\Prefx(\pi)$ for the set of strict prefixes of play~$\pi$. Such
finite prefixes are called \emph{histories}, and we~let
$\Hist_{\calG}(\pos)=\Prefx(\Play_\calG(\pos))$ and $\Hist_\calG =
\bigcup_{\pos \in \Pos} \Hist_{\calG}(\pos)$. We~extend the notion of
strict prefixes and the notation~$\Prefx$ to histories in the natural
way, requiring in particular that $\rho\notin\Prefx(\rho)$. A~(finite)
extension of a history~$\rho$ is any history~$\rho'$ such that
$\rho\in\Prefx(\rho')$.  

\smallskip
A~\emph{strategy} 
is a mapping
$\strat\colon \Hist_\calG \to \Act$, and we write $\StrSet$
for the set of strategies in $\calG$. 
An~\emph{assignment} is a
partial function $\chi\colon \Var \cup
\Agt\rightharpoonup\StrSet$, that assigns strategies to
variables and agents.
The assignment $\chi[\ag\mapsto\strat]$  maps $\ag$ to $\strat$
and is equal to $\chi$ otherwise. 
Let 
$\chi$ be an assignment and $\rho$ a history.
We define the set of  \emph{outcomes} of $\chi$ from $\rho$ as the set
$\Out(\chi,\rho)$ of plays $\pi=\rho\cdot \pos_1\pos_2\ldots$ such that
for every~${i \in \bbN}$, there exists a joint action $\jact\in\Act^\Agt$
such that for each agent $\ag\in\dom{\chi}, \jact(\ag)=\chi(\ag)(\pi_{\le
|\rho|+i-1})$ and $\pos_{i+1}=\Delta(\pos_i,\jact)$, where $\pos_0=\last(\rho)$.




\begin{definition}[Semantics]	
  Consider a WCGS
  $\calG=\tuple{\AP,\Agt,\Act,\Pos,\pos_\init,\Delta,\labels}$, a set
  of variables $\Var$, and a partial assignment $\chi$ of strategies
  for $\Agt$ and $\Var$. Given an \FSL state formula $\varphi$ and a
  history $\rho$, we use $\semSL{\chi}{\rho}{\varphi}$ to denote the
  satisfaction value of $\varphi$ in the last state of $\rho$ under
  the assignment $\chi$.  Given an \FSL path formula $\psi$, a play
  $\pi$, and a point in time $i\in\SetN$, we use
  $\semSL{\chi}{\pi,i}{\psi}$ to denote the satisfaction value of
  $\psi$ in the suffix of $\pi$ that starts in position $i$. The
  satisfaction value is defined inductively as follows:
  \begin{align*}
    \semSL{\chi}{\rho}{p} &= \labels(\last(\rho))(p) \\
    \semSL{\chi}{\rho}{\estrat\varphi} &=  \sup_{\strat \in \StrSet}
    \semSL{\chi[x \mapsto \strat]}{\rho}{\varphi} \\
    \semSL{\chi}{\rho}{(\ag,x)\varphi} & = \semSL{\chi[\ag \mapsto
      \chi(x)]}{\rho}{\varphi} \\
    \semSL{\chi}{\rho}{\A\psi} &=
    \inf_{\pi\in\Out(\chi,\rho)}\semSL{\chi}{\pi,|\rho|-1}{\psi} \\
    \semSL{\chi}{\rho}{f(\varphi_{1}, \ldots, \varphi_{m})} & =
    f(\semSL{\chi}{\rho}{\varphi_{1}}, \ldots,
    \semSL{\chi}{\rho}{\varphi_{m}}) \\
    \semSL{\chi}{\pi,i}{\varphi} &= \semSL{\chi}{\pi_{\le i}}{\varphi} \\
    \semSL{\chi}{\pi,i}{\X \psi} &= \semSL{\chi}{\pi,i+1}{\psi} \\
    \semSL{\chi}{\pi,i}{\psi_1\U\psi_2} & = \sup_{j \ge i}\min \left
      (\semSL{\chi}{\pi,j}{\varphi_2},\min_{k \in [i,j-1]}
      \semSL{\chi}{\pi,k}{\varphi_1}\right ) \\
    \semSL{\chi}{\pi,i}{f(\psi_{1}, \ldots, \psi_{m})} &= f(\semSL{\chi}{\pi,i}{\psi_{1}}, \ldots, \semSL{\chi}{\pi,i}{\psi_{m}})
	\end{align*}
\end{definition}
\bmcomment{please check} Strategy quantification $\estrat \varphi$
computes the optimal value a choice of strategy for variable $x$ can
give to formula $\varphi$. Dually, $\astrat \varphi$ computes the
minimal value a choice of strategy for variable $x$ can give to
formula $\varphi$. Binding $(a,x) \varphi$ just assigns strategy given
by $x$ to agent $a$. Temporal modality $\X\psi$ takes the value of
$\psi$ at the next step, while
$\psi_1 \U \psi_2$ maximises, over all positions along the play, the minimum 
between the value of $\psi_2$ at that position and the minimal value
of $\psi_1$ before this position. 

\bmcomment{please check} In a Boolean setting, we recover the standard semantics of
\SL. Also the fragment of \FSL with only temporal operators and
functions $\vee$ and $\neg$ corresponds to Fuzzy Linear-time Temporal Logic~\cite{lamine2000using,frigeri2014fuzzy}. 
Note that by equipping $\Func$ with adequate functions, we can capture
various classic fuzzy interpretations of boolean operators, such as
the Zadeh, G\"odel-Dummett or \L{}ukasiewicz interpretations
(see for instance~\cite{frigeri2014fuzzy}). However the interpretation
of the temporal operators is fixed in our logic.

\begin{remark}
  \label{rem-max}
  As we shall see, once we fix a finite set of possible satisfaction values for the
  atomic propositions in a formula $\varphi$, as is the case when a
  model is chosen, the set of possible satisfaction values for subformulas of
  $\varphi$ becomes finite. Therefore, the infima and suprema in the
  above definition are in fact minima and maxima.
\end{remark}

For a state formula $\varphi$ and a weighted game structure $\calG$,
we write $\semSLb{\varphi}$ for $\semSL{\emptyset}{\pos_\init}{\varphi}$.



\subsection{Model checking}
\label{sec-def-mc}


The problem we are interested in is the following generalisation of
the model checking problem, which is solved in~\cite{ABK16}
for \LTLf and \CTLsf.

\begin{definition}[Model-checking problem]
  \label{def-mc-sl}
  Given an \SLf state formula $\varphi$, a WCGS $\Game$ and
  a predicate $\pred\subseteq [0,1]$, decide whether $\semSLb{\varphi}\in\pred$.
\end{definition}

Note that $\pred$ should be finitely represented, typically
as a threshold or an interval. 

The precise complexity of the model-checking problem will be stated in
terms of \emph{nesting depth} of formulas, which counts the maximal number
  of strategy quantifiers in a formula $\varphi$, and is written $\ndepth(\varphi)$.
We establish the following result in
Section~\ref{sec-mc-slf-atl}:
\begin{theorem}
  \label{theo-mc-slf}
  The model-checking problem for \SLf is decidable. \bmchanged{It is  \kEXPTIME[(k+1)]-complete for formulas of nesting depth at most $k$.}
\end{theorem}



\subsection{What can \FSL express?}
\label{ssec-expressiveness}

\FSL naturally embeds \SL. Indeed, if the values of the atomic
propositions are in $\{0,1\}$ and the only allowed functions in $\Func$
are $\vee, \wedge$, and $\neg$, then the satisfaction value of the
formula is in $\{0,1\}$ and coincides with the value of the
corresponding \SL formula. Below we~illustrate how quantities enable
the specification of rich strategic properties.
  
  
  \paragraph*{Drone battle}

A ``carrier'' drone $c$ helped
by a ``guard'' drone $g$ try to bring an artefact to a rescue point 
and keep it away from the ``villain'' adversarial drone $v$.  
They evolve in a three
dimensional cube of side length $1$ unit, \bmchanged{in which
  coordinates are triples $\coor=(\coora,\coorb,\coorc)\in [0,1]^3$}.  We  use 
the triples of atomic propositions $p_\coor=(p_\coora, p_\coorb, p_\coorc)$  and $q_\coor=(q_\coora, q_\coorb, q_\coorc)$ to denote
the coordinates of $c$ and $v$, respectively.  
Write $\dist\colon [0,1]^3\times[0,1]^3\to [0,1]$ for the (normalized)
distance between two points in the cube.
%
%
  Let
        the atomic proposition $\safe$ denote that the artefact has
        reached the rescue point.  In \FSL, \bmchanged{we can express the level
        of safety for the artefact defined as the minimum
        distance between the carrier and the villain along a
        trajectory to
        reach the rescue point.}  Indeed, the formula
	\[
	\varphi_{\text{rescue}} = \estrat[x] \estrat [y]
        (c, x)(g, y)
        \;\A (\dist(p_\coor,q_\coor) \U \; \safe)
	\]
	states that the carrier and guard drones cooperate to keep
        the villain as far as possible from the artefact, until it is rescued.  Note that the satisfaction value of the \FLTL
        specification is~$0$ if there is a path in which the artefact is never rescued. 

        \bmchanged{The strategies of the carrier and the guard being
          quantified before that of the villain implies that they are
          unaware of the villain's future moves. Now assume the
          guard is a double agent to whom the villain communicates
          his plan.}
	\bmchanged{Then his strategy can depend on the villain's strategy, which is
        captured by the following formula:} 
	\[
	\varphi_{\text{spy}} = \estrat[x] \astrat[z] \estrat [y] (c, x)(g, y)(v, z)\;\A(\dist(p_\coor,q_\coor) \U \; \safe)\text{}
	\]
	
	Note that the formula $\varphi_{\text{rescue}}$ can be written in \FATL
	\gpcomment{Should we introduce \FATL or give just an
	intuition?}\bmcomment{We don't have room to introduce it formally}, whereas $\varphi_{\text{spy}}$ requires \FSL.
In~fact
        $\varphi_{\text{spy}}$ actually belongs to the fragment
        \FSLOneG, which we study in Section~\ref{sec-mc-sloneg}.


\paragraph*{Synthesis with quantitative objectives}
	The problem of synthesis for \LTL specifications dates back to~\cite{PR89}.
	The setting is simple: two agents, a controller and an environment, operate on two disjoint sets of variables in the system.
	The controller wants a given \LTL specification $\psi$ to be satisfied in the infinite execution, while the environment wants to prevent it.
	The problem consists into synthesising a strategy for the controller such that, no matter the behaviour of the environment, the resulting execution satisfies $\psi$.
	Recently, this problem has been addressed in the context of \FLTL, where the controller aims at maximising the value of an \FLTL formula $\varphi$, while the environment acts as minimiser.
	Both problems can be easily represented in \SL and \FSL respectively, with the formula	
	\[
	\varphi_{\synt} = \estrat \astrat[y] \; (c, x) (e, y) \;\A \psi
	\]
	
	where $c$ and $e$ are the controller and environment agent, respectively, and $\psi$ the temporal specification expressed in either \LTL or \FLTL.
	
	Assume now that controller and environment are both composed of
        more than one agent, namely $c_1, \ldots, c_n$ and
        $e_1, \ldots, e_n$, and each controller component has the power
        to adjust its strategic choice based on the strategies
        selected by the environmental agents of lower rank.  That is,
        the strategy selected by agent $c_k$ depends, on the
        strategies selected by agents $e_j$, for every $j < k$.  We
        can write a \FSL formula to represent this generalised
        synthesis problem as follows:
	\[
	\varphi_{\synt} = \estrat[x_1]  \astrat[y_1] \cdots \estrat[x_n]  \astrat[y_n] \; (c_1, x_1) (e_1, y_1) \ldots (c_n, x_n) (e_n, y_n)\;\A \psi.
	\]
	
	Notice that every controller agent is bound to an existentially quantified variable, that makes it to maximise the satisfaction value of the formula in its scope.
	On the other hand, the environmental agents are bound to a universally quantified variable, that makes them to minimise the satisfaction value.
	
Notice that in general each alternation between
        existential and universal quantification yields an additional
        exponential in the complexity of the model-checking problem,
        as we show in  the overall quantification alternates
        from existential to universal $2n-1$ times, which would induce
        a .  In
        section~\ref{sec-mc-sloneg}, we show that, for the special
        case of these formulas, such alternation does not affect the
        computational complexity of the model-checking problem.
  
\paragraph*{NE in weighted games}
  
An important feature of \SL in terms of expressiveness is that it captures
Nash equilibria (NE, for short) and other common solution
concepts. This extends to~\FSL, but in a much stronger sense: first,
objectives in \FSL are quantitative, so that \emph{profitable
  deviation} is not a simple Boolean statement; second, the semantics
of the logic is quantitative, so that \emph{being a NE} is a
quantitative property, and we can actually express \emph{how far} a
strategy profile is from being a NE.

Consider a strategy profile $(x_i)_{\ag_i\in\Agt}$.  Assuming all
agents follow their strategies in that profile, a NE can
be characterised by the fact that all agents play one of
their best responses against their opponents' strategies. We~would
then write
\[
\varphi_{\NE}  ((x_i)_{\ag_i\in\Agt}) =
(\ag_1,x_1)\ldots (\ag_n,x_n)
\bigwedge_{\ag_i\in\Agt} \estrat[y_i]\; ((\ag_i,y_i)\;\A\varphi_i) \preceq \A\varphi_i
\]
where $\alpha\preceq\beta$ equals~$1$ if $\alpha\leq\beta$ and~zero
otherwise.
Strategy profile $(x_i)_{\ag_i\in\Agt}$ is a NE if, and only~if, $\varphi_{\NE} ((x_i)_{\ag_i\in\Agt})$
evaluates to~$1$. 

Adopting a more quantitative point of view, we can measure how much agent~$i$
can benefit from a selfish deviation using formula $\estrat[y_i] \diff((\ag_i,y_i)\varphi_i, \varphi_i)$, where $\diff(x,y) = \max\{0, x-y\}$.
The maximal benefit that some
agent may get is then captured by the following formula:
\[
\varphi_{\bar\NE} ((x_i)_{\ag_i\in\Agt}) =
\estrat[y]\;
(\ag_1,x_1)\ldots (\ag_n,x_n)
\bigvee_{\ag_i\in\Agt} \diff((\ag_i,y)\A\varphi_i, \A\varphi_i).
\]
Formula~$\varphi_{\bar\NE}$ can be used to characterise
$\varepsilon$-NE, by requiring that $\varphi_{\bar\NE}$
has value less than or equal to~$\varepsilon$; of~course it also
characterises classical NE as a special case.

\paragraph*{Secure equilibria in weighted games}

Secure equilibria~\cite{chatterjee2006games} are special kinds of NEs
in two-player games, where besides improving their objectives, the
agents also try to harm their opponent. Following the ideas above,
we~characterise secure equilibria in~\FSL as follows:
\[
\varphi_{\SE}(x_1,x_2) =
  (\ag_1,x_1)(\ag_2,x_2)
\bigwedge_{i\in\{1,2\}} \estrat[y] ((\ag_i,y)\A\varphi_1,(\ag_i,y)\A\varphi_2) \preceq_i
 (\A\varphi_1,\A\varphi_2)
\]
where $(\alpha_1,\alpha_2)\preceq_i(\beta_1,\beta_2)$ is~$1$ when
$(\alpha_i\leq \beta_i) \lor (\alpha_i=\beta_i \land \alpha_{3-i}\leq
\beta_{3-i})$, and~$0$ otherwise.

Secure equilibria have also been studied in $\mathbb Q$-weighted
games~\cite{BMR14}: in that setting, the objective of the
agents is to optimise e.g. the (limit) infimum or supremum of the
sequence of weights encountered along the play.
We~can characterise secure equilibria in such setting (after first
applying an affine transformation to have all weights in~$[0,1]$):
indeed, assuming that weights are encoded as the value of atomic
proposition~$w$, the value of formula~$\G w$ is the infimum of the
weights, while the value of $\F\G w$ is the limit infimum.  We~can
then characterise secure equilibria with (limit) infimum and supremum
objectives by using those formulas as the objectives for the agents
in formula~$\varphi_{\SE}$.

Other classical properties of games can be expressed, such as doomsday
equilibria (which generalise winning secure equilibria in $n$-player
games)~\cite{CDFR14}, robust NE (considering profitable
deviations of coalitions of agents)~\cite{Bre16}, or strategy
dominance and admissibility~\cite{Ber07,BRS17}, to~cite a~few.

\paragraph*{Rational synthesis}

Weak rational synthesis~\cite{FKL10,KPV16,AKP18} aims at synthesising a strategy profile for a controller~$C_0$ and the $n$ components~$(C_i)_{1\leq i\leq n}$
constituting the environment, in such a way that (1)~the~whole system
satisfies some objective~$\varphi_0$, and (2)~under the strategy of the
controller, the~strategies of the~$n$ components form an NE (or any given solution concept) for their own individual
objectives~$(\varphi_i)_{1\leq i\leq n}$. \gpcomment{Should we describe it in the quantitative setting instead?}

That a given strategy profile $(x_i)_{C_i\in\Agt}$ satisfies the two
conditions above can be expressed as follows:
\begin{align*}
  \varphi_{\textsf{wRS}} ((x_i)_{\bmchanged{0\le i\le n}})  &=
  (\bmchanged{C_0}, x_0)(\bmchanged{C_1}, x_1) \ldots (\bmchanged{C_n}, x_n)
  [\A\varphi_0 \wedge \varphi_{\NE}((x_i)_{1\leq i\leq n})]
\end{align*}

The formula returns the minimum  between the satisfaction value of
$\varphi_0$ and that of
$\varphi_{\NE}((x_i)_{1\leq i\leq n})$. 
Thus, the satisfaction value of $\varphi_{\textsf{wRS}}$ is~zero if the strategy
profile~$(x_i)_{1\leq i\leq n}$ is not a NE under
strategy~$x_0$ assigned to~$C_0$, and it returns the satisfaction value of~$\varphi_0$
under the whole strategy profile otherwise. Then the value of
\[
\estrat[x_0] \estrat[x_1] \ldots \estrat[x_{n}]\;
\varphi_{\textsf{wRS}} ((x_i)_{\bmchanged{0\le i \le n}}) 
\]
is the best value of~$\varphi_0$ that the system can collectively
achieve under the condition that the components in the environment are
in an NE.  
Obviously, we can go beyond NE and use any other solution
concept that can be expressed in~\FSL.


The counterpart of weak rational synthesis is strong rational synthesis, that aims at synthesising a strategy only for controller $C_0$ in such a way that the objective $\varphi_0$ is maximised over the worst NE that can be played by the environment component over the strategy of $C_0$ itself.

This can be expressed as follows:
\begin{align*}
\varphi_{\textsf{sRS}} (x_0)  &=
\astrat[x_1] \ldots \astrat[x_{n}]\ (C_0, x_0)(C_1, x_1) \ldots (C_n, x_n)
[\neg \varphi_{\NE}((x_i)_{1\leq i\leq n}) \vee \varphi_0]
\end{align*}

The formula in the scope of the quantifications and bindings returns
the maximum value between $\neg \varphi_{\NE}((x_i)_{1\leq i\leq n})$
and $\varphi_{0}$.  Given that the former is $1$ if there is no NE and
$0$ otherwise, the disjunction takes value $1$ over the path with no
NEs and the value of $\varphi_{0}$, otherwise.  Given that the
environment components have universally quantified strategies, the
formula $\varphi_{\textsf{sRS}}$ amounts at minimising such
disjunction.  Thus, the components will select (if any) the NE that
minimises the satisfaction value of $\varphi_{0}$.  Then the value of
\[
\estrat[x_0] \varphi_{\textsf{sRS}} (x_0) 
\]
is the best value of $\varphi_0$ that the controller can achieve under the condition that the components in the environment are playing the Nash Equilibrium that worsens it.




\paragraph*{Core equilibria}
\gpcomment{I suggest we remove this for the IJCAI submission}
In cooperative game theory, \emph{core equilibrium} is probably the
best known solution concept and sometimes related to the one of Nash
Equilibrium for noncooperative games.  Differently from NEs (but
similarly to Strong NEs) it accounts multilateral deviations (also
called coalition deviations) that, in order to be beneficial, must
improve the payoff of the deviating agents \emph{no matter} what is
the reaction of the opposite coalition.  More formally, for a given
strategy profile $(x_i)_{\ag_i \in \Agt}$, a coalition
$C \subseteq \Agt$ has a \emph{beneficial deviation}
$(y_i)_{\ag_i \in C}$ if, for all strategy profiles
$(z_i)_{\ag_i \in \Agt \setminus C}$ it holds that
$(x_i)_{\ag_i \in \Agt}\varphi_{i} \prec (y_i)_{\ag_i \in
  C}(z_i)_{\ag_i \in \Agt \setminus C}\varphi_{i}$, for every
$\ag_i \in C$.  We say that a strategy profile
$(x_i)_{\ag_i \in \Agt}$ is a \emph{core equilibrium} if, for every
coalition, there is no beneficial deviation.  This can be written in
\FSL as follows:

\begin{align*}
\varphi_{\textsf{core}} (x_i)_{\ag_i \in \Agt}  &=
\bigwedge_{C \subseteq \Agt}
\astrat[y_i]_{\ag_i \in C}
\estrat[y_i]_{\ag_j \in \Agt \setminus C}
(\bigwedge_{\ag_j \in C} (\ag_i, y_i)_{\ag_i \in \Agt} \varphi_j \preceq (\ag_i, x_i)_{\ag_i \in \Agt} \varphi_{j})
\end{align*}

The strategy profile $(x_i)_{\ag_i \in \Agt}$ is a core equilibrium
if, and only if, the formula $\varphi_{\textsf{core}}(x_i)_{\ag_i \in
  \Agt}$ evaluates to $1$.  The existence of a core equilibrium could
be than expressed with the formula $ \estrat[x_1] \ldots \estrat[x_n]
\varphi_{\textsf{core}} (x_i)_{\ag_i \in \Agt}$, which takes value $1$
if and only if there exists a core equilibrium.


\section{Booleanly Quantified \CTLsf}
\label{sec-qctlsf}
In this section we introduce Booleanly Quantified \CTLsf (\BQCTLsf,
for short) which extends \bmchanged{both \CTLsf and \QCTLs~\cite{DBLP:journals/corr/LaroussinieM14}.
On the one hand, it extends \CTLsf with second order quantification over
Boolean atomic propositions, on the other hand it extends \QCTLs to
the quantitative setting of \CTLsf}. While \BQCTLsf formulas are interpreted
over weighted Kripke structures, thus with atomic propositions having
values in $[0,1]$, the possible assignment for the quantified
propositions are Boolean.

\subsection{Syntax}
Let $\Func\subseteq \{f\colon [0,1]^m\to [0,1] \mid m \in \SetN \}$ be
a set of functions over $[0,1]$. 
\begin{definition}
The
syntax of \BQCTLsf is defined with respect to a finite set  \AP of atomic
propositions, using the following grammar:
\begin{align*}
  \varphi&\coloncolonequals p  \mid  \existsp[p] \varphi \mid \E \psi \mid f(\varphi,\ldots,\varphi)\\
  \psi&\coloncolonequals  \varphi \mid \X \psi \mid  \psi \U \psi \mid f(\psi,\ldots,\psi)
\end{align*}
where $p$ ranges over~$\AP$ and $f$ over $\Func$.  
\end{definition}
 
Formulas of type $\varphi$ are called \emph{state formulas}, those of
type $\psi$ are called \emph{path formulas}, and \BQCTLsf consists of
all the state formulas defined by the grammar.  An atomic proposition
which is not under the scope of a quantification is called
\emph{free}. If no atomic proposition is free in a formula $\varphi$,
then we say that $\varphi$ is \emph{closed}.
We again use
$\top$, $\vee$, and $\neg$ to denote functions $1$, $\max$ and $1 -
x$, as well as classic abbreviations already introduced for \SLf, plus
$\A\psi\colonequals \neg \E \neg \psi$.


\subsection{Semantics}

\BQCTLsf formulas
are evaluated on unfoldings of weighted Kripke structures.
Note that
the terminology Boolean only concerns the quantification of atomic
propositions (which is restricted to Boolean atomic propositions), and
that formulas are interpreted over weighted Kripke structures.

\begin{definition}
  \label{def-wks}
  A \emph{weighted Kripke structure} (WKS) is a tuple $\KS =
  \tuple{\AP,\setstates,\state_\init,\relation,\lab}$ where $\AP$ is a
  finite set of atomic propositions, $\setstates$
  is a finite set of states, $\state_\init \in \setstates$ is an
  initial state, $\relation\subseteq\setstates\times\setstates$ is a
  left-total transition relation\footnote{\textit{i.e.}, for all
    $\state\in\setstates$, there exists $\state'$ such that
    $(\state,\state')\in\relation$.}, and $\lab\colon \setstates \to
  [0,1]^{\AP}$ is a weight function.
\end{definition}
A \emph{path} in $\KS$ is an infinite word
$\spath=\spath_0\spath_1\ldots$ over $S$ such that
$\spath_0=\state_\init$ and $(\spath_i,\spath_{i+1})\in \relation$ for
all $i$.  
By analogy with concurrent game structures we call finite prefixes of
paths \emph{histories}, and write $\Hist_\KS$ for the set of all histories in $\KS$. 
We also let $\setval{\KS}=\{\lab(\state)(p)\mid \state\in\setstates \mbox{
  and }p\in\AP\}$ be the finite set of values appearing in $\KS$.

\subparagraph*{Trees}
Given finite sets $D$ of directions, $\AP$ of atomic propositions, and
$V \subseteq [0,1]$ of possible values, an $(\AP,V)$-labelled $D$-tree,
(or tree for short when the parameters are understood or irrelevant),
is a pair $\ltree=(\tree,\lab)$ where $\tree\subseteq D^+$ is closed
under non-empty prefixes, all \emph{nodes} $\node\in\tree$ start with
the same direction $\root$, called the \emph{root}, and have at least
one \emph{child} $\node\cdot d\in \tree$, and $\lab:\tree\to V^{\AP}$
is a \emph{weight function}.  
A
\emph{branch} $\tpath=\node_0\node_1\ldots$ is an infinite sequence of
nodes such that for all $i \geq 0$, we have that $\node_{i+1}$ is a
child of $\node_i$. We let $\Br{\node}$ be the set of branches that
start in node $\node$.  Given a tree $\ltree=(\tree,\lab)$ and a node
$\node \in \tree$, we define the \emph{subtree of $\ltree$ rooted in
  $\node$} as the tree $\ltree_\node = (\tree_\node,\lab')$ where
$\tree_\node = \{v \in \setstates^+: u \pref v\}$ ($\pref$ denotes the
non-strict prefix relation) and $\lab'$ is $\lab$ restricted to
$\tree_\node$.
We say that a tree $\ltree=(\tree,\lab)$ is \emph{Boolean in $p$},
written  $\boolp{\ltree}{p}$, if
for all $\node\in\tree$ we have $\lab(\node)(p)\in\{0,1\}$.
As with weighted Kripke structures, we let
$\setval{\ltree}=\{\lab(\node)(p)\mid \node\in\tree \mbox{ and
}p\in\AP\}$.

\smallskip
Given two $(\AP,V)$-labelled $D$-trees $\ltree,\ltree'$ and $p \in
\AP$, we write $\ltree \Pequiv[p] \ltree'$ if $t$ and $t'$ differ only in
assignments to $p$, which \bmchanged{must be Boolean in $\ltree'$}: formally, $t= (\tree,\lab)$,
$t'=(\tree,\lab')$, for the same domain $\tree$, 
\bmchanged{$\boolp{\ltree'}{p}$}, and for all $p' \in \AP$ such that $p' \neq p$
and all $\node \in \tree$, we have $\lab(\node)(p')=
\lab'(\node)(p')$.
\bmcomment{I removed the constraint $\boolp{\ltree}{p}$ as $p$ may be
  part of the model and have non-boolean value}


Finally, we define the \emph{tree unfolding of a weighted Kripke
  structure $\KS$} over atomic propositions $\AP$ and states
$\setstates$ 
  as the $(\AP,\setval{\KS})$-labelled $\setstates$-tree $\unfold{} =
  (\Hist_\KS,\lab')$, where $\lab'(\node) = \lab(\last(\node))$ for
  every $\node \in \Hist_\KS$. 

\begin{definition}[Semantics]
  \label{def-sem-qctl}
  Consider finite sets $D$ of directions, $\AP$ of atomic
  propositions, and $V \subseteq [0,1]$ of possible values. We fix an
  $(\AP,V)$-labelled $D$-tree $\tree$. Given a \BQCTLsf state formula
  $\varphi$ and a node $\node$ of $\tree$, we use
  $\sem{\node}{\varphi}$ to denote the satisfaction value of $\varphi$
  in node $\node$. Given a \BQCTLsf path formula $\psi$ and a branch
  $\tpath$ of $\tree$, we use $\sem{\tpath}{\psi}$ to denote the
  satisfaction value of $\psi$ along $\tpath$.  The satisfaction value
  is defined inductively as follows:
  \begin{align*}
    \sem{\node}{p} &= \lab(\node)(p)\\ 
    \sem{\node}{\existsp \varphi} &= \sup_{\ltree' \mid \ltree'\Pequiv[p] \ltree}
                                    \,\sem[\ltree']{\node}{\varphi}\\ 
    \sem{\node}{\E\psi} &= \sup_{\tpath\in\Br{\node}}
                          \,\sem{\tpath}{\psi}\\ 
    \sem{\node}{f(\varphi_1,\ldots,\varphi_n)} &= f(\sem{\node}{\varphi_1},\ldots,\sem{\node}{\varphi_n})\\ 
  \sem{\tpath}{\varphi} &= \sem{\tpath_0}{\varphi}\ \text{where}\
  \tpath_0\ \text{is the first node of}\ \tpath\\ 
  \sem{\tpath}{\X \psi} &= \sem{\tpath_{\geq 1}}{\psi} \\
  \sem{\tpath}{\psi_1 \U \psi_2} &= \sup_{i\geq
                                   0} \min(\sem{\tpath_{\geq
                                   i}}{\psi_2}, \min_{0\leq j <i}
                                   \sem{\tpath_{\geq j}}{\psi_1})\\
    \sem{\tpath}{f(\psi_1,\ldots,\psi_n)} &= f(\sem{\tpath}{\psi_1},\ldots,\sem{\tpath}{\psi_n})
\end{align*}  
\end{definition}
\pbcomment{should we redefine $\lambda_{\ge i}$ here?} \bmcomment{I
  don't think it is necessary}

\begin{remark}
  \label{rem-max-qctl}
  As with~\FSL, we will see that the suprema in the above definition
  can be replaced with maxima (see Lemma~\ref{lem-exp-values} below).
\end{remark}

\bmchanged{First, we point out that if $\Func=\{\top,\vee,\neg\}$, then \BQCTLsf evaluated on boolean Kripke structures
  corresponds to classic \QCTLs~\cite{DBLP:journals/corr/LaroussinieM14}.}
Note also that the  quantifier on propositions does not
range over arbitrary values in $[0,1]$. Instead, as in  \QCTLs,
it quantifies only on Boolean valuations.  It is still quantitative
though, in the sense that instead of merely stating the existence of a
valuation, $\existsp \varphi$ maximises the value of $\varphi$ over
all possible (Boolean) valuations of $p$.



For a tree $\ltree$ with root $\root$ we write $\semb{\varphi}$ for
$\sem{\root}{\varphi}$, and for a weighted Kripke structure $\KS$ we
 write $\semb[\KS]{\varphi}$ for $\semb[\unfold{}]{\varphi}$.
Note that this semantics is an extension of the tree semantics of
\QCTLs, in which the valuation of quantified atomic propositions is
chosen on the unfolding of the Kripke structure instead of the
states. This allows us to capture the semantics of Strategy Logic
based on strategies with perfect recall, where moves can depend on the
history, as apposed to the memoryless semantics, where strategies can
only depend on the current state
(see~\cite{DBLP:journals/corr/LaroussinieM14} for more detail).


As for \SLf, we are interested in the following model-checking problem:

\begin{definition}
  \label{def-mc-qctl}
  Given a \BQCTLsf state formula $\varphi$, a weighted Kripke structure $\KS$, and
  a predicate $\pred\subseteq [0,1]$, decide whether $\semb[\KS]{\varphi}\in\pred$.
\end{definition}

Similarly to \FSL, the precise complexity of the model-checking problem will be stated in
terms of \emph{nesting depth} of formulas, which counts the maximal number
  of nested propositional quantifiers in a formula $\varphi$, and is written $\ndepth(\varphi)$.

In the next section we establish our
main technical contribution, which is the following:

\begin{theorem}
  \label{theo-mc-qctlsf}
  The  quantitative model-checking problem for \BQCTLsf is decidable.
\bmchanged{It is  \kEXPTIME[(k+1)]-complete for formulas of nesting depth at most $k$.}
\end{theorem}


This result, together with a reduction from \SLf to \BQCTLsf that we
present in Section~\ref{sec-mc-slf-atl}, entails the decidability
of model checking \SLf announced in Theorem~\ref{theo-mc-slf}.


\section{Model checking \BQCTLsf}
\label{sec-mc-qctlsf}
We start by proving that, as has been the case for \LTLf, since the
set of possible satisfaction values of an atomic
proposition is finite, so is the set of satisfaction values of each \BQCTLsf formula.
This property allows to use $\max$
instead of $\sup$ in Definition~\ref{def-sem-qctl}.


\begin{lemma}
  \label{lem-exp-values}
  Let $\setval{}\subset [0,1]$ be a finite set of values with
  $\{0,1\}\subseteq \setval{}$, let $\varphi$ be a \BQCTLsf state
  formula and $\psi$ a \BQCTLsf path formula, with respect to
  $\AP$. Define
  \[
  \setval{\varphi}=\{\sem{\node}{\varphi}\mid \ltree\ \text{is a
    $(\AP,V)$-labelled tree}\ \text{and}\ \node\in\ltree\}
  \]
  be the set of values taken by $\varphi$ in nodes of
  $(\AP,V)$-labelled trees.  Similarly, define
  \[
  \setval{\psi}=\{\sem{\tpath}{\psi}\mid \ltree\ \text{is a
    $(\AP,V)$-labelled tree},\ \node\in\ltree\ \text{and}\
  \tpath\in\Br{\node}\}
  \] 
  Then, $|\setval{\varphi}|\le |\setval{}|^{|\varphi|}$ and
  $|\setval{\psi}|\le |\setval{}|^{|\psi|}$.  Moreover, one can
  compute 
  sets $\approxsetval{\varphi}$ and $\approxsetval{\psi}$ such that
  $\setval{\varphi} \subseteq \approxsetval{\varphi}$ and
  $\setval{\psi} \subseteq \approxsetval{\psi}$
  of size at most $|\setval{}|^{|\varphi|}$ and $|\setval{}|^{|\psi|}$, respectively.
\end{lemma}

\begin{proof}
  We prove the result by mutual induction on $\varphi$ and
  $\psi$.
 Clearly, 
 $\setval{p}=\setval{}$. 

For $\varphi=\existsp \varphi'$, 
observe that if $\setval{\ltree}\subseteq \setval{}$ and $u\in\ltree$, then for all
 trees $\ltree'$ such that $\ltree'\Pequiv[p] \ltree$  
 it is also the case that $\setval{\ltree'}\subseteq\setval{}$
 (by assumption $\setval{}$ contains 0 and 1). 
 It follows that  $\sem{\node}{\existsp \varphi'}$ is
 defined as the supremum of a subset of $\setval{\varphi'}$, 
 which by induction hypothesis is of size at most
 $|\setval{}|^{|\varphi'|}$, and thus the supremum is indeed a
 maximum. It follows that $\sem{\node}{\existsp
   \varphi'}\in\setval{\varphi'}$. Hence,
 $\setval{\existsp\varphi'}\subseteq\setval{\varphi'}$, and thus
 $|\setval{\existsp\varphi'}|\le |\setval{\varphi'}|\le
 |\setval{}|^{|\varphi'|}\le |\setval{}|^{|\existsp\varphi'|}$.

 For $\varphi=\E\psi$, again $\sem{\node}{\E\psi}$ is a supremum over
 a subset of $\setval{\psi}$, which by induction hypothesis is of size
 at most $\setval{}^{|\psi|}$. The supremum is thus reached, hence
 $\setval{\E\psi}\subseteq\setval{\psi}$
and  $|\setval{\E\psi}|\le |\setval{\psi}|\le
 |\setval{}|^{|\psi|}\le |\setval{}|^{|\E\psi|}$.

 For $\varphi=f(\varphi_1,\ldots,\varphi_n)$, we have
 $\setval{\varphi}=\{f(v_1,\ldots,v_n)\mid v_i\in
 \setval{\varphi_i}\}$,
 hence
 $|\setval{\varphi}|$ is at most
 $\prod_{i=1}^n|\setval{\varphi_i}|$. By induction hypothesis, we
get $|\setval{\varphi}|\le\prod_{i=1}^n|\setval{}|^{|\varphi_i|}\le
|\setval{}|^{|\varphi_1|+\ldots + |\varphi_n|}\le
|\setval{}|^{|\varphi|}$.

For $\psi=\varphi$, the result follows by  hypothesis of mutual
induction.

For $\psi=\X\psi'$, we have that $\setval{\psi}=\setval{\psi'}$, and the result
follows.

For $\psi=\psi_1\U\psi_2$, the value of $\psi$ is defined via
suprema and infima over possible values for $\psi_1$ and $\psi_2$,
which are finitely many by the induction hypothesis. The suprema and
infima are thus maxima and minima, and
$\setval{\psi}\subseteq\setval{\psi_1}\cup\setval{\psi_2}$. Hence,
$|\setval{\psi}|\le |\setval{\psi_1}|+|\setval{\psi_2}|\le
|\setval{}|^{|\psi_1|}+|\setval{}|^{|\psi_2|}\le
|\setval{}|^{|\psi_1|+|\psi_2|}\le |\setval{}|^{|\psi|}$ (since
$|\setval{}| \ge 2$).

In all cases, the claim for over-approximations follows by the same
reasoning as above.
\end{proof}


The finite over-approximation of the set of possible satisfaction
values induces a finite alphabet for the automata our model-checking
procedure uses.

In the following, we use \emph{alternating parity tree automata}
(APT in short), and their purely non-deterministic (resp. universal)
variants, denoted NPT (resp. UPT). Given two APT $\auto$ and $\auto'$
we denote $\auto\wedge\auto'$ (resp. $\auto\vee\auto'$) the APT of
 size linear in $|\auto|$ and $|\auto'|$ that accepts the
 intersection (resp. union) of the languages of $\auto$ and
 $\auto'$, and we call \emph{index} of an automaton the number of
 priorities in its parity condition. We refer the reader to
 \bmcomment{what to cite?} for a detailed exposition of alternating
 parity tree automata.

We extend the automata-based model-checking procedure for \CTLs from
\cite{KVW00}. Note that since the quantified atomic propositions may
appear in different subformulas, we cannot extend the algorithm for
\CTLsf from~\cite{ABK16}, as the latter applies the technique of~\cite{EL87}, where the evaluation of each subformula is independent.

\begin{proposition}
  \label{prop-mc-qctl}
  Let $\values\subset[0,1]$ be a finite set of values such that
  $\{0,1\} \subseteq \setval{}$, and let $D$ be a finite set of
  directions. For every $\BQCTLsf$ \bmcomment{no, we need to construct it for open
    formulas too} state formula $\varphi$ and predicate $\pred\subseteq
  [0,1]$, one can construct an APT $\automc{\pred}{\varphi}$ such that
  for every $(\AP,\values)$-labelled $D$-tree $\ltree$, 
  $\automc{P}{\varphi}$ accepts $\ltree$ if and only if
  $\semb{\varphi}\in\pred$.

  The APT $\automc{\pred}{\varphi}$ has at most
  $(\ndepth(\varphi)+1)$-exponentially many states, and its index is at most
  $\ndepth(\varphi)$-exponential. 
\end{proposition}

\begin{proof}
  The proof proceeds by an induction on the structure of the formula
  $\varphi$ and strengthens the induction statement as follows: one
  can construct an APT $\automc{\pred}{\varphi}$ such that for every
  $(\AP,\values)$-labelled $D$-tree $\ltree$, for every node $\node
  \in\ltree$, we have that $\automc{P}{\varphi}$ accepts $\ltree$ from
  node $\node$ if and only if $\sem{\node}{\varphi} \in \pred$.

  If $\varphi=p$, the automaton has one state and accepts a tree
  $\ltree=(\tree,\lab)$ in node $\node\in\tree$ if
  $\lab(\node)(p)\in\pred$, and rejects otherwise.  In addition,
  $\setval{p}=\setval{}$.

  If $\varphi=\existsp \varphi'$, we want to check whether the maximal
satisfaction  value of $\varphi'$ for all possible Boolean valuations of $p$ is in
  $P$. To do so we first compute a finite set
  $\approxsetval{\varphi'}$ of exponential size such that
  $\setval{\varphi'}\subseteq \approxsetval{\varphi'}$, which we can
  do as established by Lemma~\ref{lem-exp-values}. For each possible
  value $v\in \approxsetval{\varphi'}\cap P$, we check whether this
  value is reached for some $p$-valuation, and if the value of
  $\varphi'$ is less than or equal to $v$ for all $p$-valuations.  For
  each $v\in\approxsetval{\varphi'}\cap P$, inductively build the APTs
  $\automc{\{v\}}{\varphi'}$ and $\automc{[0,v]\cap P}{\varphi'}$.
  Turn the first one into a NPT~$\Nauto_{=v}$
  and the second one into a UPT~$\Uauto_{\le v}$. Project $\Nauto_{=v}$
  existentially on~$p$, and call the result $\Nauto'_{=v}$. Project
  $\Uauto_{\le v}$ universally on $p$, call the result $\Uauto'_{\le
    v}$. Finally, we can define the APT $\automc{P}{\existsp
    \varphi'}\egdef \bigvee_{v\in \approxsetval{\varphi'}\cap P}
  \Nauto'_{=v}\wedge \Uauto'_{\le v}$. It~is then easy to see that
  this automaton accepts a tree if and only if there exists a value in
  $P$ that is the maximum of the possible values taken by $\varphi'$
  for all $p$-valuations.

  If $\varphi=\E \psi$: as in the classic automata construction for
  \CTLs~\cite{DBLP:journals/jacm/KupfermanVW00}, we first let
  $\atoms(\psi)$ be the set of maximal state sub-formulas of $\psi$
  (that we call \emph{atoms} thereafter~--~which have to be
  distinguished from atomic propositions of the formula). In a first
  step we see elements of $\atoms(\psi)$ as atomic propositions, and
  $\psi$ as an \LTLf formula over $\atoms(\psi)$. According to
  Lemma~\ref{lem-exp-values} we can compute over-approximations
  $\approxsetval{\varphi'}$ for each $\varphi'\in\atoms(\psi)$, and we
  thus let
  $\approxsetval{}=\bigcup_{\varphi'\in\atoms(\psi)}\approxsetval{\varphi'}$
  be a finite over-approximation of the set of possible values for
  atoms.
  It is proven in~\cite{ABK16} that for every $\pred\subseteq [0,1]$,
  one can build a nondeterministic \bmchanged{parity} automaton 
  $\wauto_\pred^{\psi}$ \bmchanged{of size exponential in $|\psi|^2$} such that $\wauto_\pred^{\psi}$ accepts a word
  $w\in (\approxsetval{}^{\atoms(\psi)})^\omega$ if and only if
  $\semc{w}{\psi}\in \pred$. Now let us compute $\approxsetval{\E\psi}$
  (again using Lemma~\ref{lem-exp-values}), and for each
  $v\in\approxsetval{\E\psi}\cap \pred$, construct an NPT
  $\Nauto^{E=v}$ that guesses a branch in its input and simulates
  $\wauto_{\{v\}}^\psi$ on it.
\bmchanged{To obtain a universal word automaton of single exponential
  size that checks whether $\semc{w}{\psi}\in [0,v]$, first build the
  nondeterministic automaton $\wauto^\psi_{(v,1]}$ from~\cite{ABK16},
and  dualize it in linear time. From the resulting universal automaton $\wauto_{[0,v]}^{\psi}$
we build a UPT $\Uauto^{A\le v}$}
  that executes $\wauto_{[0,v]}^{\psi}$ on all branches of its
  input.\footnote{We take $\wauto_{[0,v]}^{\psi}$ universal because it
    is not possible to simulate a nondeterministic word automaton on
    all branches of a tree, but it is possible for a universal
    one. \bmcomment{cite something?} 
    Note that we could also determinise $\wauto_{[0,v]}^{\psi}$, but
    it would cost one more exponential.} We now define
  the APT $\auto^\pred$ on $\approxsetval{}^{\atoms(\psi)}$-trees as
  \[\auto^\pred= \bigvee_{v\in \approxsetval{\E\psi}\cap
    \pred}\Nauto^{E=v}\wedge \Uauto^{A\le v}.\]

  Now to go from atoms to standard atomic propositions,
  we define an APT $\automc{P}{\E
    \psi}$ that simulates $\auto^\pred$ 
  by, in each state and
  each node of its input, guessing a value $v_i$ in
  $\approxsetval{\varphi_i}$ 
  for each formula $\varphi_i\in\atoms(\psi)$, simulating
  ${\auto^\pred}$ on the resulting label, and launching a copy
  of $\auto_{\varphi_i}^{V,\{v_i\}}$ for each
  $\varphi_i\in\atoms(\psi)$. \bmchanged{Note that the automaton is alternating
  and thus may have to guess several times the satisfaction value of a
 formula $\varphi_i$ in a same node, but launching the
 $\auto_{\varphi_i}^{V,\{v_i\}}$ forces it to always guess the same,
 correct value.}




  Finally, if $\varphi=f(\varphi_1,...,\varphi_n)$,  \bmchanged{we list
  all combinations $(v_1,\ldots,v_n)$ of the possible satisfaction values for the
  subformulas $\varphi_i$ such that $f(v_1,\ldots,v_n)\in\pred$, and
we  build  automaton~$\automc{\pred}{\varphi}$ as the disjunction over
such $(v_1,\ldots,v_n)$ of the conjunction of
  automata $\automc{v_i}{\varphi_i}$.}

  \smallskip
  The complexity of this procedure is non-elementary. More precisely,
  we claim that $\auto_{\varphi}^{V,P}$ has size at most
  $(\ndepth(\varphi)+1)$-exponential and index (i.e., number of
  priorities for the parity condition) at most $\ndepth(\varphi)$-exponential.

  The case where~$\varphi$ is an atomic proposition is trivial.

  For~$\varphi=\existsp \varphi'$, we~transform an
  exponential number of APTs into NPTs or UPTs. This entails an
  exponential blowup in the size and index of each
  automaton. The~resulting automaton $\automc{P}{\existsp \varphi'}$
  has  at most $(\ndepth(\varphi')+2)$-exponentially many states and index at most
  $(\ndepth(\varphi')+1)$-exponential. Since $\ndepth(\varphi)=\ndepth(\varphi')+1$,
  the inductive property is preserved.

  If $\varphi$ 
  has the form $f(\varphi_1,...,\varphi_n)$, then the automaton
  for~$\varphi$ is a combination of the automata for all~$\varphi_i$,
  and for the various values those subformulas may take. By
  Lemma~\ref{lem-exp-values} there are at
  most $|\values|^{|\varphi_1|+\ldots+|\varphi_n|}\le |\values|^{|\varphi|}$
  different combinations, so assuming (from the
  induction hypothesis) that the automata for~$\varphi_i$ have 
  at most $(\ndepth(\varphi_i)+1)$-exponentially many states and index at most
  $h(\varphi_i)$-exponential, the automaton for~$\varphi$ has 
  at most $(\ndepth(\varphi)+1)$-exponentially many states and index at most
  $\ndepth(\varphi)$-exponential (note that $\ndepth(\varphi_i)=\ndepth(\varphi)$).

Finally for $\varphi=\E\psi$: following~\cite{ABK16}, the size of
  $\wauto_\pred^{\psi}$ is exponential in~$|\psi|^2$, and 
\bmchanged{at most $|\psi|$ B\"uchi acceptance conditions. One can turn this automaton into an equivalent
  B\"uchi automaton still  exponential in $|\psi|^2$, which can
  be seen as a parity automaton with index $2$.} Then
  $\Nauto^{E=v}$ and~$\Uauto^{A\le v}$ both also have sizes
  exponential in~$|\psi|^2$, and index  $2$.
  Finally, $\auto^\pred$, which combines an exponential number of the
  automata above, has size exponential in~$|\psi|^2$ and
  index~$2$. 
  The~final
  automaton $\automc{P}{\E \psi}$ is obtained from that automaton by
  plugging the automata for $\auto_{\varphi_i}^{V,\bullet}$, whose
  sizes and indices are dominating the size and index
  of~$\auto^\pred$. It follows that, for~$\varphi=\E\psi$, the
  size of~$\automc{P}{\varphi}$ also is $(\ndepth(\varphi)+1)$-exponential, and
  its index is $\ndepth(\varphi)$-exponential.
  \end{proof}




To see that Theorem~\ref{theo-mc-qctlsf} follows from
Proposition~\ref{prop-mc-qctl}, recall that by
definition $\semb[\KS]{\varphi}=\semb[\unfold{}]{\varphi}$. Thus to
check whether $\semb[\KS]{\varphi}\in\pred$, where atoms in $\KS$ takes values
in $\setval{}$, it is enough to build $\automc{\pred}{\varphi}$ as in
Proposition~\ref{prop-mc-qctl}, take its product with a deterministic
tree automaton that accepts only~$\unfold{}$, and check for emptiness
of the product automaton.
The formula complexity is $(\ndepth(\varphi)+1)$-exponential, but the
structure complexity is polynomial.

\bmchanged{For the lower bounds, consider the fragment \EQkCTLs of \QCTLs
 which  consists in formulas in
        prenex normal form, \ie with all
         quantifications on atomic propositions at the beginning,  with at most
        $k$ alternations between existential and universal
        quantifiers, counting the first quantifier as one alternation (see~\cite[p.8]{DBLP:journals/corr/LaroussinieM14} for
        a formal definition). Clearly, \EQkCTLs can be translated in
        \BQCTLsf with formulas of linear size and nesting depth at most
        $k$ (alternation is simply coded by placing function $\neg$
        between quantifiers).
        It is proved in~\cite{DBLP:journals/corr/LaroussinieM14} that model
        checking \EQkCTLs is \kEXPTIME[(k+1)]-hard.        
}



\section{Model checking quantitative strategic logics}
\label{sec-mc-slf-atl}


In this section we show how to reduce the model-checking problem for \FSL to
that of \QCTLsf. This reduction is a rather  straightforward adaptation
of the usual one for qualitative variants of \SL (see
\eg~\cite{DBLP:journals/iandc/LaroussinieM15,DBLP:conf/lics/BerthonMMRV17,DBLP:conf/csl/FijalkowMMR18}). 
We essentially observe that it can be lifted to the quantitative
setting.


We let $\Agt$ be a finite set of agents, and $\AP$ be a finite set of
atomic propositions.

\subparagraph*{Models transformation.}  We first define for every WCGS
$\Game = \tuple{\Act,\Pos,\pos_\init,\Delta,\labels}$ over $\Agt$ and
$\AP$ a WKS $\KS_{\Game} =
\tuple{\setstates,\state_\init,\relation,\lab}$ over some set $\AP'$
and a bijection $\fplay\mapsto\node_{\fplay}$ between the set of
histories starting in the initial state $\pos_\init$ of $\Game$ and
the set of nodes in $\unfold[\KS_{\Game}]{}$.  We consider
propositions $\AP_\Pos=\{p_{\pos}\mid\pos\in\Pos\}$, that we assume to
be disjoint from $\AP$.  We let $\AP' = \AP \cup \AP_\Pos$. Define the
Kripke structure
$\KS_{\Game}=(\setstates,\state_{\init},\relation,\lab)$ where
\begin{itemize}
\item $\setstates=\{\state_{\pos} \mid \pos\in\Pos\}$, 
\item $\state_{\init}=\state_{\pos_{\init}}$,
\item $\relation=\{(\state_{\pos},\state_{\pos'}) \in \setstates^2\mid
\exists\jact\in\Act^{\Agt} \mbox{ s.t. }\Delta(\pos,\jact)=\pos'\}$, and
\item $\lab(p)(\state_{\pos})=
  \begin{cases}
    1 & \mbox{if }p\in \AP_\Pos \mbox{ and }p=p_\pos\\
    0 & \mbox{if }p\in \AP_\Pos \mbox{ and }p\neq p_\pos\\
    \labels(p)(\pos) & \mbox{otherwise}
  \end{cases}$.
\end{itemize}
For every history $\fplay=\pos_{0}\ldots\pos_{k}$, define the node
$\node_{\fplay}= \state_{\pos_{0}}\ldots \state_{\pos_{k}}$ in
$\unfold[\KS_{\Game}]{}$ (which exists, by definition of $\KS_{\Game}$
and of tree unfoldings). Note that the mapping
$\fplay\mapsto\node_{\fplay}$ defines a bijection between the set of
histories from $\pos_\init$ and the set of nodes in
$\unfold[\KS_{\Game}]{}$.

\subparagraph*{Formulas translation.}  Given a game
$\Game=\tuple{\Act,\Pos,\pos_\init,\Delta,\labels}$ and a formula
$\varphi\in\FSL$, we define a \QCTLsf formula $\varphi'$ such that
$\semSLb{\varphi}=\semb[\KS_\Game]{\varphi'}$. More precisely, this
translation is parameterised with a partial function $g:\Agt
\rightharpoonup \Var$ which records bindings of agents to strategy
variables.  Suppose that $\Act=\{\act_{1},\ldots,\act_{\maxmov}\}$.
We define the translation $\trs[g]{\cdot}$ 
by 
induction on 
state formulas $\varphi$ and path formulas $\psi$.  Here is the definition of
$\trs[g]{\cdot}$ for state formulas:
\[
\begin{array}{lr}
\trs[g]{p}  = p \quad\quad\quad\quad& 
\trs[g]{(\ag,\var)\varphi} = \trs[{g[\ag\mapsto \var]}]{\varphi} \\  
\multicolumn{2}{l}{\trs[g]{\estrat\varphi} = \exists
  p_{\act_{1}}^{\var}\ldots \exists
  p_{\act_{\maxmov}}^{\var}. \Big(\phistrat \wedge \trs[g]{\varphi}\Big), }\\[10pt]
\multicolumn{2}{r}{\mbox{where }
\phistrat  =
\A\G\Big(\bigvee_{\act\in\Act}(p_{\act}^{\var}\wedge\bigwedge_{\act'\neq\act}\neg
p_{\act'}^{\var})\Big)}\\
\multicolumn{2}{l}{  
  \trs[g]{\A\psi} = \A(\psiout[g] \rightarrow \trp[g]{\psi})}\\[10pt]
\multicolumn{2}{r}{  
\quad\quad\quad\mbox{where }
\psiout[g] = \G \left(
\bigwedge_{\pos\in\Pos} \left( p_{\pos} \rightarrow
\bigvee_{\jact\in\Act^{\Agt}} 
\left ( \bigwedge_{\ag\in\dom{f}}p_{\jact(\ag)}^{f(\ag)} \wedge \X
            p_{\Delta(\pos,\jact)} \right ) \right )\right)}\\[10pt]
  \multicolumn{2}{l}{
  \trs[g]{f(\varphi_1,\ldots,\varphi_n)}  =
                                       f(\trs[g]{\varphi_1},\ldots,\trs[g]{\varphi_n})}            
\end{array}
\]
and for path formulas:
\begin{align*}
\trp[g]{\varphi} & = \trs[g]{\varphi} &
\trp[g]{\X\psi} 		& = \X\trp[g]{\psi} \\
\trp[g]{\psi\U\psi'}	& =
\trp[g]{\psi}\U\trp[g]{\psi'} &
  \trp[g]{f(\psi_1,\ldots,\psi_n)} & =
                                       f(\trp[g]{\psi_1},\ldots,\trp[g]{\psi_n})
\end{align*}

This translation is identical to that from branching-time \SL to
\QCTLs in all cases, except for the case of functions which is
straightforward. To see that it can be safely lifted to the
quantitative setting, it suffices to observe the following: since
quantification on atomic propositions is restricted in \BQCTLsf to
Boolean values, and atoms in $\AP_\Pos$ also have Boolean values,
$\phistrat$ and $\psiout$ always have value 0 or 1 and thus they can
play exactly the same role as in the qualitative setting: $\phistrat$
holds if and only if the atomic propositions
$p_{\act_1}^\var,\ldots,p_{\act_m}^\var$ indeed code a strategy from
the current state, and $\psiout$ holds on a branch of
$\unfold[{\KS_{\Game}}]{}$ if and only if in this branch each agent
$\ag\in\dom{g}$ follows the strategies coded by atoms $p_\act^\var$.
As a result $\exists p_{\act_{1}}^{\var}\ldots \exists
p_{\act_{\maxmov}}^{\var}. \Big(\phistrat \wedge
\trs[g]{\varphi}\Big)$ maximises over those valuations for the
$p_{\act_i}^\var$ that code for strategies, other valuations yielding
value 0. Similarly, formula $\A(\psiout[g] \rightarrow \trp[g]{\psi})$
minimises over branches that represent outcomes of strategies in $g$,
as others yield value 1.

One can now see that the following holds, where $\varphi$ is an \SLf formula.
\begin{lemma}
\label{lem-redux}
Let $\assign$ be an assignment and $g:\Agt\rightharpoonup \Var$ such that $\dom{g}=\dom{\assign}\cap \Agt$ and for all $a \in \dom{g}$,
$g(\ag) = \var$ implies $\assign(\ag) = \assign(\var)$. Then
\[\semSL{\assign}{\rho}{\varphi}=\sem[{\unfold[\KS_\Game]{}}]{\node_{\rho}}{\trs[g]{\varphi}}\]
\end{lemma}

As a result, the quantitative model-checking problem for an \SLf
formula $\varphi$, a weighted CGS $\Game$ and a predicate
$\pred\subseteq[0,1]$ can be solved by computing the \BQCTLsf formula
$\varphi'=\trs[\emptyset]{\varphi}$ and the weighted Kripke structure
$\KS_\Game$, and deciding whether
$\semb[\KS_\Game]{\varphi'}\in\pred$, which can be done by
Theorem~\ref{theo-mc-qctlsf}. \bmchanged{This establishes  the upper-bounds in
Theorem~\ref{theo-mc-slf}.
 As in the case of \BQCTLsf, the lower-bounds are obtained by reduction from
the model-checking problem for \EQkCTLs. This reduction is an
adaptation of the one from \QCTLs to \ATL with
strategy context in~\cite{DBLP:journals/iandc/LaroussinieM15}, and
that preserves nesting depth.}

\begin{remark}
  \label{rem-finitely-many-values-slf}
  Lemma~\ref{lem-redux} together with Lemma~\ref{lem-exp-values} imply
  that \FSL formulas also can take only exponentially many values when
  a finite domain is fixed for atomic propositions. This justifies the
  observation of Remark~\ref{rem-max} that supremum and infimum in the
  semantics of \FSL can be replaced with maximum and minimum.
\end{remark}





\section{The case of \FSLOneG}
\label{sec-mc-sloneg}
We now study the fragment \FSLOneG, which is the extension to the
quantitative setting of the one-goal fragment \SLOneG of
\SL~\cite{MMPV12}. 

In order to define the syntax, we need to introduce the notions
of \emph{quantification prefix} and \emph{binding prefix}.  A
quantification prefix is a sequence
$\qntpref = \qnt{x_{1}} \ldots \qnt{x_{n}}$, where
$\qnt{x_{i}} \in \{\estrat[x_{i}], \astrat[x_{i}]\}$ is either an
existential or universal quantification.  For a fixed set of agents
$\Agt = \{a_{1}, \ldots, a_{n}\}$ a binding prefix is a sequence
$\bndpref = (a_{1}, x_{1}), \ldots (a_{n}, x_{n})$, where every agent
in $\Agt$ occurs exactly once.  A combination $\qntpref \bndpref$ is
\emph{closed} if every variable occurring in $\bndpref$ occurs in some
quantifier of $\qntpref$.  We can now give the definition of \FSLOneG
syntax.

\begin{definition}[\FSLOneG Syntax]
  Let $\AP$ be a set of Boolean atomic propositions, and let $\Func
  \subseteq \{f:[0, 1]^m \to [0, 1] \mid m \in \SetN \}$ be a set of functions over $[0,1]$.  
  The set of \FSLOneG formulas is defined by the following grammar:
  \begin{align*}
    \varphi & := p \mid f(\varphi, \ldots, \varphi) \mid \varphi \mid
    \X \varphi \mid \varphi \U \varphi \mid \qntpref \bndpref \varphi,
  \end{align*}
  where $p\in\AP$, $f\in \Func$, and $\qntpref \bndpref$ is a closed combination of a quantification prefix and of a binding prefix.
\end{definition}

Note that all \FSL formulas are sentences, as all  strategy
variables are quantified immediately before being bound to some agent.
The \emph{sentence nesting depth} of an \FSLOneG formula is defined as follows:
\begin{itemize}
\item $\SentDpth(p) = 0$ for every $p \in \AP$;
\item $\SentDpth(f(\varphi_{1}, \ldots, \varphi_{n})) = \max_{1 \leq i
    \leq n}\{\SentDpth(\varphi_{i})\}$;
\item $\SentDpth(\X \psi) = \SentDpth(\psi)$;
\item $\SentDpth(\psi_1 \U \psi_2) = \max \{\SentDpth(\psi_1),
  \SentDpth(\psi_2)\}$.
\item $\SentDpth(\qntpref \bndpref \varphi) = \SentDpth(\varphi) + 1$;
\end{itemize}
Intuitively, the sentence nesting depth measures the number of
sentences, i.e., formulas with no free agent or variable, that are
nested into each other in the formula.

In order to solve the model-checking problem for \FSLOneG, we 
need the technical notion of \emph{concurrent multi-player parity
  game} introduced in~\cite{MMS16}.
\begin{definition}
  \label{def:cmpg}
  A \emph{concurrent multi-player parity game} (CMPG) is a tuple $\Parity = \tuple{\Agt, \Act, \setstates, \state_{\init}, \parFun, \Delta}$, where $\Agt = {0, \ldots, n}$ is a set of agents indexed with natural numbers, $\setstates$ is a set of states, $\state_{\init}$ is a designated initial state, $\parFun: \setstates \to \SetN$ is a priority function, and $\Delta: \setstates \times \Act^{\Agt} \to \setstates$ is a transition function determining the evolution of the game according to the joint actions of the players.
\end{definition}

A CMPG is a game played by players $\Agt = {0, \ldots, n}$ 
 for an infinite number of rounds.  In
each round, the players concurrently and independently choose moves,
and the current state and the action chosen for each player determine
the successor state.  In details we have that each player $i$, with $i
\mod 2 = 0$ is part of the existential (even) team; the other players
are instead part of the universal (odd) team.  Informally, the goal in
a CMPG is to check whether there exists a strategy for $0$ such that,
for each strategy for $1$, there exists a strategy for $2$, and so
forth, such that the induced plays satisfy the parity condition.
Then, we say that the existential team wins the game.  Otherwise the
universal team wins the game.

As shown in~\cite[Theorem 4.1, Corollary 4.1]{MMS16}, one can decide
the winners of a CMPG $\Parity = \tuple{\Agt, \Act, \setstates,
  \state_{\init}, \parFun, \Delta}$ in time polynomial
w.r.t. $\card{\setstates}$ and $\card{\Act}$, and exponential
w.r.t. $\card{\Agt}$ and $k=\max \parFun$ (the maximal priority).

%

\begin{theorem}
  \label{thm:mcatlsf}
  The model-checking problem for closed formulas of \FSLOneG is
  decidable, and \twoexpC.
\end{theorem}


\begin{proof}
  We let $\calG=\tuple{\AP,\Agt,\Act,\Pos,\pos_\init,\Delta,\labels}$
  be a WCGS and we consider a formula of the form $\qntpref \bndpref
  \varphi$.
  We also assume, for simplicity, that $\qntpref = \estrat[x_0] \astrat[x_1], \ldots, \qnt{x_k}$, that is, quantifiers perfectly alternate between existential and universal.\footnote{To reduce to this case, one can either collapse the agents occurring in the same quantification block, or interleave them with dummy agents quantified with the other modality.}
  Note that the formula $\qntpref \bndpref \varphi$ is a sentence, therefore the choice of an assignment is useless.
  Moreover, recall that, by Lemma~\ref{lem-redux} and in particular Remark~\ref{rem-finitely-many-values-slf}, the set $V(\qntpref \bndpref \varphi)$ of possible values is bounded by $2^{\card{\qntpref \bndpref \varphi}}$.

  We proceed by induction on the sentence nesting depth.  As base case
  let $\SentDpth(\qntpref \bndpref \varphi) = 1$, i.e., there is no
  occurrence of neither quantifiers nor bindings in $\varphi$.  Then,
  $\varphi$ can be regarded as an \FLTL formula that can be
  interpreted over paths of the WKS $\KS =
  \tuple{\AP,\Pos,\pos_\init,\relation,\lab}$ where $\relation =
  \{(\pos_1,\pos_2) \mid \exists \jact \in \Act^\Agt .\ \pos_2 =
  \Delta(\pos_1,\jact)\}$.
  Now, thanks to~\cite[Theorem 3.1]{ABK16}, for every value $v \in
  V(\qntpref \bndpref \varphi)$, there exists a nondeterministic
  generalised B\"{u}chi word automaton
  $\mthname{B}_{\Kripke, \varphi, P_v}$, with $P_v = [v, 1]$ that
  accepts all and only the infinite paths $\pi$ of $\KS$
  such that $\sem[\KS]{\pi}{\varphi} \in P_v$.
  \pbcomment{formulation not very consistent with the proof of QCTLf where we speak of infinite words over $V^\Act$}
  \gpcomment{Is the problem that words recognised by the automaton are not on the alphabet $V^\Act$?}
  Following~\cite{Piterman07}, we can convert $\mthname{B}_{\Kripke, \varphi, P_v}$ into a deterministic parity word automaton $\Automaton_{\Kripke, \varphi, P_v} = \tuple{\Pos, Q, q_{\init}, \delta, \parFun}$ of size doubly-exponential in the size of $\varphi$ and index bounded by $2^{\card{\varphi}}$.
	
  At this point, define the following CMPG $\Parity = \tuple{\Agt',
    \Act, \setstates, \state_\init, \parFun', \Delta'}$ such that
  \begin{itemize}
  \item $\Agt' = \{0, \ldots, k\}$ is a set of agents, one for every
    variable occurring in $\qntpref$, ordered in the same way as in
    $\qntpref$ itself;
  \item $\Act$ is the set of actions in $\Game$;
  \item $\setstates = \Pos \times Q$ is the product of the states of
    $\Game$ and the automaton $\Automaton_{\Kripke, \varphi, P_v}$;
  \item $\state_{\init} = (\pos_\init, q_\init)$ is the pair given by
    the initial states of $\Game$ and $\Automaton_{\Kripke, \varphi,
      t}$, respectively;
  \item $\parFun'(\pos, q) = \parFun(q)$ mimics the parity function of
    $\Automaton_{\Kripke, \varphi, P_v}$;
  \item if $\jact \in \Act^{\Agt'}$, $\Delta'((\pos, q), \jact) =
    (\Delta(\pos, \bndpref(\jact)), \delta(q,\pos))$ mimics the
    execution of both $\Game$ and $\Automaton_{\Kripke, \varphi, P_v}$.
  \end{itemize}

  The game emulates two things, one per each component of its
  state-space.  In the first, it emulates a path $\pi$ generated in
  $\Game$.  In the second, it emulates the execution of the automaton
  $\Automaton_{\Kripke, \varphi, P_v}$ when it reads the path $\pi$
  generated in the first component.  By construction, it results that
  every execution $(\pi, \eta) \in \Pos^\omega \times Q^\omega$ in
  $\Parity$ satisfies the parity condition determined by $\parFun'$
  if, and only if, $\sem[\KS]{\pi}{\varphi} \in P_v$.  Moreover,
  observe that, since $\Automaton_{\Kripke, \varphi, P_v}$ is
  deterministic, for every possible history $\rho$ in $\Game$, there
  is a unique partial run $\eta_{\rho}$ that makes the partial
  execution $(\rho, \eta_\rho)$ possible in $\Parity$.  This makes the
  sets of possible strategies $\StrSet(\pos_\iota)$ and
  $\mthset{Str}_{\Parity}(\state_{\iota})$ in perfect bijection.
  $\Parity$ has a winning strategy if and only if
  $\semSL{}{\pos_\iota}{\qntpref \bndpref \varphi} \in P_v$.  In order
  to compute the exact value of $\semSL{}{\pos_\iota}{\qntpref
    \bndpref \varphi}$, we repeat the procedure described above for
  every $v \in V(\qntpref \bndpref \varphi)$ and take the maximum $v$
  of those for which $\semSL{}{\pos_\iota}{\qntpref \bndpref \varphi}
  \in P_v$.
	
  For the induction case, assume we can compute the satisfaction value
  of every \FSLOneG formula with sentence nesting depth at most $n$,
  and let $\SentDpth(\qntpref \bndpref \varphi) = n + 1$.  Observe
  that, for every subsentence $\qntpref' \bndpref' \varphi'$ of
  $\qntpref \bndpref \varphi$, we have that $\SentDpth(\qntpref'
  \bndpref' \varphi') \leq n$ and so, by induction hypothesis, we can
  compute $\semSL{}{\pos}{\qntpref' \bndpref' \varphi'}$ for every $\pos \in
  \Pos$.  Now, introduce a fresh atomic proposition $p_{(\qntpref'
    \bndpref' \varphi')}$ whose weight in $\Game$ is defined as
  $\labels(v)(p_{(\qntpref' \bndpref' \varphi')}) =
  \semSL{}{v}{\qntpref' \bndpref' \varphi'}$ and a set of fresh atomic
  propositions $p_\pos$, one for every $\pos \in \Pos$, whose weights in
  $\Game$ are defined as $\labels(\pos)(p_{\pos}) = 1$ and
  $\labels(\pos)(p_{\pos'}) = 0$ if $\pos \ne \pos'$.
  Now, consider the Boolean formula
  $$
  \Phi(\qntpref' \bndpref' \varphi') = \bigvee_{\pos \in \Pos} (p_\pos \wedge p_{(\qntpref' \bndpref' \varphi')})
  $$
	
  Observe that every disjunct is a conjunction of the form $p_\pos
  \wedge p_{(\qntpref' \bndpref' \varphi')}$, whose satisfaction value
  on a state $\pos'$ is the minimum among the weights of $p_{\pos}$
  and $p_{(\qntpref' \bndpref' \varphi')}$.  This can be either $0$,
  if $\pos \neq \pos'$ or $\semSL{}{\pos}{\qntpref' \bndpref'
    \varphi'}$, if $\pos = \pos'$.  Now, the big disjunction takes the
  maximum among them.  Therefore, we obtain that
  $\semSL{}{\pos}{\Phi(\qntpref' \bndpref' \varphi')} =
  \semSL{}{\pos}{\qntpref' \bndpref' \varphi'}$.  This allows us to
  replace every occurrence of $\qntpref' \bndpref' \varphi'$ in
  $\varphi$ with the Boolean combination $\Phi(\qntpref' \bndpref'
  \varphi')$, making the resulting formula to be of sentence nesting
  depth $1$.  Thus, we can apply the procedure described in the base
  case, to compute $\semSL{}{\pos}{\qntpref \bndpref \varphi}$.

  Regarding the complexity, note that the size of $\Parity$ is
  $\card{\Pos} \cdot \card{Q}$ that is in turn linear with respect to
  the size of $\Game$ and doubly exponential in the size of $\varphi$.
  This is due to the fact that the automaton $\Automaton_{\Kripke,
    \varphi, P_v}$ results from the construction of the NGBW
  $\mthname{B}_{\Kripke, \varphi, P_v}$, of size singly exponential in
  $\card{\varphi}$ and the transformation to a DPW, that adds up
  another exponential to the construction.  On the other hand, the
  number of priorities in $\Parity$ is only singly exponential in
  $\card{\varphi}$, and it is due to the fact that the transformation
  from NGBW to DPW requires a singly exponential number of priorities.
  Therefore, the CMPG $\Parity$ can be solved in time polynomial
  w.r.t. the size $\Game$ and double exponential in $\card{\varphi}$.
  Such \twoexp procedure is executed a number of time exponential in
  $\varphi$, which is still \twoexp.
	
	Hardness follows from that of \SLOneG~\cite{MMPV12}.
\end{proof}


\section{Future work}
\label{sec-future}
We introduced and studied \FSL, a~formalism for specifying quality and
fuzziness of strategic on-going behaviour. Beyond the applications
described in the paper, we highlight here some interesting directions
for future research. In~classical temporal-logic model checking,
\emph{coverage} and \emph{vacuity} algorithms measure the sensitivity
of the system and its specifications to mutations, revealing errors in
the modelling of the system and lack of exhaustiveness of the
specification \cite{CKV06}. When applied to \FSL, these algorithms can
set the basis to a formal reasoning about classical notions in game
theory, like the sensitivity of utilities to price changes, the
effectiveness of burning money \cite{HR08,SR16} or tax increase
\cite{CDR06}, and more. Recall that our \FSL model-checking algorithm
reduces the problem to \BQCTLsf, where the quantified atomic
propositions take Boolean values. It~is interesting to extend \BQCTLsf
to a logic in which the quantified atomic propositions are associated
with different agents, which would enable easy specification of
controllable events. Also, while in our application the quantified
atomic propositions encode the strategies, and hence the restriction
of their values to $\{0,1\}$ is natural, it is interesting to study
\QCTLsf, where quantified atomic propositions may take values in
$[0,1]$.





\bibliography{bibexport}

\end{document}